\documentclass{article}
\usepackage[colorlinks = true]{hyperref}
\usepackage{amssymb, amsmath, amsthm, graphicx, pagecolor, tikz, subcaption}

\DeclareMathOperator{\trace}{tr}

\newtheorem{lemma}{Lemma}
\newtheorem{theorem}{Theorem}
\newtheorem{deffinition}{Deffinition}

\newtheorem{corollary}{Corollary}

\title{Ihara Zeta Entropy}
\author{Supriyo Dutta\thanks{Email: \texttt{dosupriyo@gmail.com}}, Partha Guha\thanks{Email: \texttt{partha@bose.res.in}} \\ S. N. Bose National Centre for Basic Sciences \\ Block - JD, Sector - III, Salt Lake City, Kolkata - 700 106}
\date{} 

\begin{document}
	\maketitle
	
	\begin{abstract}
		In this article, we introduce an entropy based on the formal power series expansion of the Ihara Zeta function. We find a number of inequalities based on the values of the Ihara zeta function. These new entropies are applicable in symbolic dynamics and the dynamics of billiards. 
	\end{abstract}

	\section{Introduction}
	
		The theory of the Ihara zeta function is an interesting topic of research in graph theory and combinatorics \cite{terras2010zeta, ihara1966discrete, giscard2017algebraic, deitmar2015ihara}. Let $G = (V(G), E(G))$ be a simple, finite, connected, undirected graph without any vertex of degree one. In addition, $G$ is not a cycle graph. Throughout, this article $n$ is the number of vertices, and $m$ is the number of edges. In the literature of the Ihara zeta function, we assign two orientations for two opposite directions on every edge. Thus, the set of all orientations on the edges can be collected as 
		\begin{equation}\label{alphabet}
		\mathcal{E} = \{e^{(1)}, e^{(2)}, \dots e^{(m)}, e^{(m+1)} = (e^{(1)})^{-1}, \dots e^{(2m)} = (e^{(m)})^{-1} \}.
		\end{equation}
		If the orientation of an edge $e$ is given by $e^{(k)} = (u,v)$ then its inverse orientation is $e^{(m + k)} = (e^{(k)})^{-1} = (v,u)$ for $k = 1, 2, \dots m$. Here, $u$ is the initial vertex of $e^{(k)}$ and $v$ is the terminating vertex of $e^{(k)}$. We denote $i(e^{(k)}) = u$ and $t(e^{(k)}) = v$. A cycle in $G$ is a finite walk $W = e_1 e_2 \dots e_r$, such that, the initial vertex of $e_1$ is the terminating vertex of $e_r$. Here $r$ is the length of the cycle. In the theory of Ihara zeta function, we also assume that in an oriented cycle $e_{k +1} \neq (e_k)^{-1}$ for any $k = 1, 2, \dots (r-1)$ and $e_1 \neq (e_r)^{-1}$. Two cycles $W^{(1)} = e^{(1)}_1 e^{(1)}_2 \dots e^{(1)}_r$ and $W^{(2)} = e^{(2)}_1 e^{(2)}_2 \dots e^{(2)}_r$ are equivalent if $e^{(2)}_1 = e^{(1)}_k, e^{(2)}_2 = e^{(1)}_{k+1}, \dots e^{(2)}_{r - k + 1} = e^{(1)}_{r}, e^{(2)}_{r - k + 2} = e^{(1)}_{1}, \dots e^{(2)}_{r} = e^{(1)}_{k - 1}$ for some $k = 1, 2, \dots (r-1)$. The set of equivalence classes of cycles are called prime cycles. We denote a prime cycle by $P$ and its length is denoted ny $\gamma(P)$. Given any graph $G$ the Ihara zeta function is defined by 
		\begin{equation}\label{Ihara_zeta_1}
		\zeta_G(z) = \prod_{[P]}\left(1 - z^{\gamma(P)}\right)^{-1},
		\end{equation}
		where $z \in \mathbb{C}$ and $0 < |z| \leq \alpha_G$, which is the radius of convergence for $\zeta_G(z)$. 
		
		The idea of entropy plays a crucial role in thermodynamics, information theory, and dynamical system. The community of theoretical and mathematical physicists is interested to generalize the idea of entropy. Recently, a number of entropies based on formal group theory \cite{dieudonne1973introduction} is proposed in a series of articles \cite{tempesta2015theorem, tempesta2016beyond}. Although $\zeta_G$ is an infinite product, the function $(\zeta_G(x)) ^{-1}$ can be expressed as a power series in terms of the adjacency matrix of the oriented line graph associated to $G$ \cite{kotani20002}. In this work, we introduce the entropy functions based on the formal power series of Ihara zeta function first time in literature. We describe these entropies as Ihara entropy. This investigation gifts us a number of inequalities and interesting properties Ihara zeta function. The interface of symbolic dynamics and graph theory is well investigated. The Ihara entropy, introduced in this article is applicable in the context of the symbolic dynamics of billiards. 
		
		This article is distributed as follows. In section 2 we derive a number of properties of Ihara zeta function which are applicable for further investigations. Section 3 is dedicated to the background and definition of Ihara entropy. Entropy is a measure of information contained in a random variable. In section 4 we discuss the information theoretic properties of Ihara entropy, which includes a check of Shannon-Khinchin axioms. Interestingly the Ihara zeta function is the Rulle zeta function in the context of symbolic dynamics. Hence in section 5, we apply the Ihara entropy in the symbolic dynamics of billiard balls. Then we conclude the article.

	\section{Properties of Ihara zeta function}
		
		Recall that there are two orientations on every edge of the graph $G$. All the orientations are collected in the set $\mathcal{E}$ described in equation (\ref{alphabet}). The oriented line graph of a given graph $G$ is denoted by $\overline{G} = (V(\overline{G}), E(\overline{G}))$ which is defined as follows:
		\begin{equation}\label{oriented_line_graph}
		\begin{split}
		V(\overline{G}) & = \mathcal{E},\\
		E(\overline{G}) & = \{((u, v), (v, w )) \in \mathcal{E} \times \mathcal{E} : u \neq w; u, v, w \in V(G)\}.
		\end{split}
		\end{equation}
		Hence, the vertices in the graph $\overline{G}$ are the elements in $\mathcal{E}$. Clearly $\overline{G}$ has $2m$ vertices. Note that, $\overline{G}$ is a connected simple graph as the graph $G$ is considered as connected and simple.
		
		Recall that the degree of a vertex in a graph is the total number of edges incident to the vertex. In addition, the degree of a graph is the sum of all vertex degrees, which is twice the total number of edges in the graph. Now, we have the following observation.
		\begin{lemma}
			The oriented line graph of $G$ contains $2\sum_{(u, v) \in E(G)}(d_u + d_v - 2)$ edges, where $d_u$ and $d_v$ are the degree of the vertices $u$ and $v$ in $G$.
		\end{lemma}
		
		\begin{proof}
			In the graph $G$ there are $d_u$ edges adjacent to the vertex $u$. One of them is $(u,v)$. Thus, there are $(d_u - 1)$ edges incident to $u$ other than $(u, v)$ in $G$. In $\overline{G}$, $(u,v)$ is a vertex which is adjacent to all vertices representing an edges $(w,u) \in E(G)$. Hence, considering the vertex $u$ in graph $G$ we find $(d_u - 1)$ vertices in $\overline{G}$ which are adjacent to $(u,v)$ in $\overline{G}$. Similarly, considering the edges incident to $v$ in the graph $G$ we can conclude that there are other $(d_v - 1)$ vertices in $\overline{G}$ which are adjacent to $(u,v)$. Therefore number of vertices adjacent to $(u, v)$ in $\overline{G}$ is $(d_u + d_v - 2)$. Also, the vertex $(v, u) \in V(\overline{G})$ which is adjacent to $(d_v + d_u - 2)$ vertices. Thus total number of edges in $\overline{G}$ is $2\sum_{(u, v) \in E(G)}(d_u + d_v - 2)$.
		\end{proof}
		Recall that, the adjacency matrix $A(G) = (a_{ij})$ of a graph $G$ is a binary $(0, 1)$ matrix such that,
		\begin{equation}
		a_{i,j} = \begin{cases} 1 & \text{if}~ (i,j) \in E(G), ~\text{and}\\ 0 & \text{if}~ (i,j) \notin E(G). \end{cases} 
		\end{equation}
		Throughout this article we denote the adjacency matrix of the oriented line graph $\overline{G}$ by $T$, which is a symmetric matrix of order $2m$. Also the eigenvalues of the adjacency matrix are considered as the eigenvalues of the graph. As $\overline{G}$ is a connected simple graph, that is a strongly connected graph, the matrix $T$ is irreducible. Hence, the Perron-Frobenius theorem indicates that $T$ has a positive eigenvalue, called the dominant eigenvalue, which is greater than or equal to (in absolute value) all other eigenvalues \cite{brouwer2011spectra}. Throughout this article the dominant eigenvalue of $\overline{G}$ is denoted by $\lambda$. 
		
		\begin{lemma}
			If $d = \min_{(u,v) \in E(G)}(d_u + d_v)$. Then, $\lambda \geq \frac{2}{m}(d - 2)$.
		\end{lemma}
		
		\begin{proof}
			The total number of edges in $\overline{G}$ is 2$\sum_{(u, v) \in E(G)}(d_u + d_v - 2)$. Every edge contributes $2$ in the degree of a graph. Hence, the degree of $\overline{G}$ is $4\sum_{(u, v) \in E(G)}(d_u + d_v - 2)$. As there are $2m$ vertices in $\overline{G}$, the average degree is $\frac{2}{m}\sum_{(u, v) \in E(G)}(d_u + d_v - 2)$. Therefore, the maximum eigenvalue of $\lambda \geq \frac{2}{m}\sum_{(u, v) \in E(G)}(d_u + d_v - 2)$. As $d = \min_{(u,v) \in E(G)}(d_u + d_v)$, we have $\lambda \geq \frac{2}{m}(d - 2)$ \cite{brouwer2011spectra}.
		\end{proof}
		
		Let $D = \max\{(d_u + d_v): (u,v) \in E(G)\}$. As the maximum eigenvalue of a graph is less that or equal to the maximum degree of its vertices, we have $\lambda \leq (D-2)$ where the maximum degree of a vertex in $\overline{G}$ is $(D - 2)$. Combining all these we find
		\begin{equation}
		\begin{split} 
		& \frac{2(d - 2)}{m} \leq \lambda \leq D - 2\\
		\text{or}~ & \frac{1}{(D - 2)} \leq \frac{1}{\lambda} \leq \frac{m}{2(d - 2)}.
		\end{split} 
		\end{equation}
		
		The alternative expression of Ihara zeta function is a formal power series which is given as follows: 
		\begin{equation}\label{Ihara_zeta_2}
		\zeta_G(z) = \exp\left(\sum_{k = 1}^\infty \frac{\trace(T^k)}{k}z^k\right),
		\end{equation}
		where $|z| < \frac{1}{\lambda}$ \cite{kotani20002}. Clearly, the function $\zeta_G(z)$ is an analytic function for $|z| < \frac{1}{\lambda}$. Hence, all its derivatives exist and have convergent power series in the region of convergence. The coefficients of the power series of $\zeta_G(z)$ is investigated in \cite{bulo2012efficient}, but we have a different approach in this article.
		
		In this work, we need the values of $\zeta_G(z)$ where $z$ is a probability function that is a real number. Hence, we need a restriction of $\zeta_G(z)$ in the real interval $[0, \frac{1}{\lambda})$, that is the function $\zeta_G(x): [0, \frac{1}{\lambda}) \rightarrow \mathbb{R}$, $\zeta_G(x) = \exp\left(\sum_{k = 1}^\infty \frac{\trace(T^{k})}{k}x^{k}\right)$. As $T$ is an adjacency matrix of a simple graph $\trace(T) = 0$. Hence, 
		\begin{equation}\label{zeta_with_real_coefficients}
		\begin{split} 
		\zeta_G(x) = &\exp\left(\sum_{k = 2}^\infty \frac{\trace(T^{k})}{k}x^{k}\right) \\
		= & 1 + \sum_{k = 2}^\infty \frac{\trace(T^{k})}{k}x^{k} + \frac{1}{2!}\left(\sum_{k = 2}^\infty \frac{\trace(T^{k})}{k}x^{k}\right)^2 + \frac{1}{3!} \left(\sum_{k = 2}^\infty \frac{\trace(T^{k})}{k}x^{k} \right)^3 + \dots \\
		= & 1 + c_1x + c_2x^2 + c_3x^3 + c_4x^4 + c_5x^5 + c_6x^6 + \dots, ~\text{where}\\
		c_1 = & 0, c_2 = \frac{\trace(T^2)}{2}, c_3 = \frac{\trace{T^3}}{3}, c_4 =  \frac{\trace(T^4)}{4} + \frac{(\trace(T^2))^2}{8}\\
		c_5 = & \frac{\trace(T^2)\trace(T^3)}{6} + \frac{\trace(T^5)}{5}\\
		c_6 = & \frac{1}{144}\left[3(\trace(T^2))^2 + 8(\trace(T^2))^2 + 18(\trace(T^2))(\trace(T^4)) + 24(\trace(T^6))\right],
		\end{split} 
		\end{equation}
		and so on. As an infinite sum of positive numbers, $\zeta_G(x) \geq 0$ for all $x \in [0, \frac{1}{\lambda})$. 
		
		If $\lambda_i$ is an eigenvalue of the matrix $T$, we find $\lambda_i^k$ is an eigenvalue of $T^k$. Also, $\trace(T^k) = \sum_{i = 1}^{2m} \lambda_i^k$. As $\lambda = \max_i\{\lambda_i\} > 0$ we have $|\trace(T^k)| \leq \sum_{i = 1}^{2m} |\lambda_i^k| \leq 2m \lambda^k$. Also, $\lambda^k \leq \trace(T^k)$, when $k$ is an even number.
		
		\begin{lemma}
			The series $\zeta_G(\frac{1}{\lambda})$ of real numbers is divergent.
		\end{lemma}
		\begin{proof}
			Note that,
			\begin{equation}
			\begin{split}
			\zeta_G(\frac{1}{\lambda}) = & 1 + c_1 \frac{1}{\lambda} + c_2\frac{1}{\lambda^2} + c_3\frac{1}{\lambda^3} + c_4\frac{1}{\lambda^4} + c_5\frac{1}{\lambda^5} + c_6\frac{1}{\lambda^6} + \dots \\
			\geq & 1 + \frac{\trace(T^2)}{2 \lambda^2} + \frac{\trace{T^3}}{3 \lambda^3} + \frac{\trace(T^4)}{4\lambda^4} + \frac{\trace(T^5)}{5\lambda^5} + \frac{\trace(T^6)}{6\lambda^6} + \dots \\
			> & 1 + \frac{1}{2} + \frac{1}{4} + \frac{1}{6} + \dots ~\text{as}~  \frac{\trace(T^k)}{\lambda^k} \geq 1 ~\text{for all even number}~ k\\
			\geq & 1 + \frac{1}{2}(1+ \frac{1}{2} + \frac{1}{3} + \dots),
			\end{split}
			\end{equation}
			which is a divergent series. Hence, the proof.
		\end{proof}
		
		\begin{lemma}
			Given any graph $G$ the Ihara zeta function $\zeta_G(x)$ and its derivatives are monotone increasing function in $[0, \frac{1}{\lambda})$.
		\end{lemma}
		
		\begin{proof}         
			We can calculate from equation (\ref{zeta_with_real_coefficients}) that, the first derivative of $\zeta_G(x)$ is
			\begin{equation}\label{1_st_derivative_of_zeta}
			\begin{split}
			\zeta'_G(x) & = \zeta_G(x) \sum_{k = 2}^\infty \trace(T^{k})x^{k - 1}, \\ 
			& = \zeta_G(x)\left(\trace(T^2) x + \trace(T^3) x^2 + \trace(T^4) x^3 + \trace(T^5) x^4 + \trace(T^6) x^5 +  \dots \right).\\
			\end{split}
			\end{equation}
			Also, the second derivative of $\zeta_G(x)$ is
			\begin{equation}\label{2_nd_derivative_of_zeta}
			\begin{split} 
			\zeta''_G(x) = \zeta'_G(x) & \sum_{k = 2}^\infty \trace(T^{k})x^{k - 1} + \zeta_G(x) \sum_{k = 2}^\infty (k - 1)\trace(T^{k})x^{k - 2}\\
			= \zeta_G(x) & [\trace(T^2) + 2 x \trace(T^3) + 3 x^2 \trace(T^4) + 4 x^3 \trace(T^5) + 5 x^4 \trace(T^6) + \dots \\
			& + (x \trace(T^2) + x^2 \trace(T^3) + x^3 \trace(T^4) + x^4 \trace(T^5) + x^5 \trace(T^6) + \dots)^2].
			\end{split} 
			\end{equation}
			Clearly, $\zeta'_G(x) \geq 0, \zeta''_G(x) \geq 0$ and similarly $\zeta^{(3)}_G(x) \geq 0$. It indicates $\zeta_G(x)$, $\zeta'_G(x)$, and $\zeta''_G(x)$ are all monotone increasing functions with minimum at $x = 0$ and diverges to infinite at $x = \frac{1}{\lambda}$. The equation (\ref{zeta_with_real_coefficients}) indicates that, $\zeta_G(0) = 1$, $\zeta_G'(0) = 0$, and $\zeta_G''(0) = 2c_2 = \trace(T^2)$.
		\end{proof}
		
		\begin{lemma}\label{denominator}
			Given any graph $G$ the function $h: [0, \frac{1}{\lambda}) \rightarrow \mathbb{R}$ defined by $h(x) = 1 - x\zeta'_G(x)$ has an unique root in $(0, \frac{1}{\lambda})$.
		\end{lemma}
		
		\begin{proof}
			As $\zeta_G'(x)$ is a monotone increasing function, $h(x)$ is a monotone decreasing function. Note that, $h(0) = 1$ and $\lim_{x \rightarrow \frac{1}{\lambda}} h(x) < 0$. Also, $h(x)$ is a continuous function. Combining we can conclude that the equation $1 - x\zeta'_G(x) = 0$ has an unique solution in $(0, \frac{1}{\lambda})$. 
		\end{proof}
		
		We denote the root of $h(x) = 0$ in $(0, \frac{1}{\lambda})$ as $x_0$. The lemma \ref{denominator} indicates that
		\begin{equation}\label{values_of_denominator}
		h(x) \begin{cases} & > 0 ~\text{when}~ 0 \leq x < x_0, \\ & = 0 ~\text{when}~ x = x_0, ~\text{and} \\ & < 0 ~\text{when}~ x_0 < x < \frac{1}{\lambda}.\end{cases}
		\end{equation}
		
		\begin{lemma}\label{zeta_1_by_2m_lambda} 
			$$\zeta_G\left(\frac{1}{2m\lambda}\right) \leq \frac{(2m)^{2m}}{e(2m-1)^{2m}}.$$
		\end{lemma}
		\begin{proof}
			Recall that $\log(\frac{1}{1 - x}) = x + \frac{x^2}{2} + \frac{x^3}{3} + \dots$. From equation (\ref{zeta_with_real_coefficients}) we find that 
			\begin{equation}
			\zeta_G\left(\frac{1}{2m\lambda}\right) = \exp\left(\sum_{k = 2}^\infty \frac{\trace(T^k)}{k(2m\lambda)^k}\right) \leq \exp\left(\sum_{k = 2}^\infty \frac{2m}{k(2m)^k}\right) ~\text{since}~ \trace(T^k) < 2m \lambda^k.
			\end{equation}
			Expanding we find
			\begin{equation}
			\begin{split}
			\zeta_G\left(\frac{1}{2m\lambda}\right) &\leq \exp\left[2m\left\{\frac{1}{2} \left(\frac{1}{2m}\right)^2 + \frac{1}{3} \left(\frac{1}{2m}\right)^3 + \frac{1}{4} \left(\frac{1}{2m}\right)^4 + \dots \right\}\right]\\
			&\leq \exp\left[2m\left\{-\frac{1}{2m} - \log\left(1 - \frac{1}{2m}\right)\right\}\right]\\
			& \leq \exp\left[\log\left\{\frac{(2m)^{2m}}{e(2m-1)^{2m}}\right\}\right] = \frac{(2m)^{2m}}{e(2m-1)^{2m}}
			\end{split}
			\end{equation}
		\end{proof}
		
		\begin{lemma}\label{function_f}
			The function $f:\mathbb{R}^+ \rightarrow \mathbb{R}$ defined by $f(x) = \frac{x^x}{e(x-1)^{x+1}}$ is a monotone decreasing function.
		\end{lemma}
		\begin{proof}
			We have 
			\begin{equation}
			f'(x) = -\frac{1}{e}(x-1)^{-x-2} x^x ((x-1) \log (x-1)-x \log (x)+\log (x)+2) < 0,
			\end{equation} 
			for $x \ge 1$. Therefore, $f(x)$ is a monotone decreasing function.
		\end{proof}
		`    
		\begin{lemma}\label{trace_T_k_2m_lambda_k}
			$$\sum_{k = 2}^\infty \frac{\trace(T^k)}{(2m\lambda)^k} \leq \frac{1}{2m-1}.$$
		\end{lemma}
		\begin{proof}
			As $\trace(T^k) \leq 2m \lambda^k$ we have $\frac{\trace(T^k)}{(2m\lambda)^k} \leq \frac{2m \lambda^k}{(2m)^k\lambda^k} \leq \frac{1}{(2m)^{k-1}}$. Therefore, 
			\begin{equation}
			\begin{split}
			\sum_{k = 2}^\infty \frac{\trace(T^k)}{(2m\lambda)^k} & \leq \sum_{k = 2}^\infty \frac{1}{(2m)^{k-1}} = \left[\frac{1}{2m} + \frac{1}{(2m)^2} + \frac{1}{(2m)^3} + \dots\right] \\
			& = \frac{1}{2m} \left[1 + \frac{1}{2m} + \frac{1}{(2m)^2} + \dots \right] \leq \frac{1}{2m} \frac{1}{1 - \frac{1}{2m}} = \frac{1}{2m-1}.
			\end{split}
			\end{equation}
		\end{proof}
		
		\begin{lemma}\label{h_1_by_2m_lambda_greater_than_0} 
			$$h\left(\frac{1}{2m\lambda}\right) > 0,$$
			where the function $h(x)$ is defined in Lemma \ref{denominator}.
		\end{lemma}
		\begin{proof}
			Proving $h\left(\frac{1}{2m\lambda}\right) > 0$ is equivalent to prove $\frac{1}{2m\lambda}\zeta_G'(\frac{1}{2m \lambda}) < 1$. The equation (\ref{1_st_derivative_of_zeta}) leads us to write 
			\begin{equation}
			\begin{split}
			\frac{1}{2m\lambda}\zeta_G'\left(\frac{1}{2m \lambda}\right) & \leq \frac{1}{2m \lambda} \zeta_G\left( \frac{1}{2m \lambda} \right) \sum_{k = 2}^\infty \trace(T^k) \left(\frac{1}{2m \lambda} \right)^{k - 1}\\
			& \leq \zeta_G\left( \frac{1}{2m \lambda} \right) \sum_{k = 2}^\infty \left(\frac{\trace(T^k)}{2m \lambda} \right)^{k - 1} \leq \frac{1}{2m-1} \zeta_G\left( \frac{1}{2m \lambda} \right),
			\end{split}
			\end{equation} 
			from lemma \ref{trace_T_k_2m_lambda_k}. Now from lemma \ref{zeta_1_by_2m_lambda} we find that,
			\begin{equation}
			\frac{1}{2m\lambda}\zeta_G'\left(\frac{1}{2m \lambda}\right) \leq \frac{1}{2m-1} \frac{(2m)^{2m}}{e(2m-1)^{2m}} = \frac{(2m)^{2m}}{e(2m-1)^{2m+1}}.
			\end{equation}
			Now lemma \ref{function_f} states that $f(2m) = \frac{(2m)^{2m}}{e(2m-1)^{2m+1}}$. For $m = 2$, we have $f(4) \approx .38756$. Also, $f$ is a monotone decreasing function. Combining all these statements we get $\frac{1}{2m\lambda}\zeta_G'\left(\frac{1}{2m \lambda}\right) < 1$. Hence, the proof.
		\end{proof}
		
		\begin{lemma}\label{zeta_1_by_2m_lambda_greater_than_2}
			$$\zeta_G \left(\frac{1}{2m \lambda}\right) > 2.$$
		\end{lemma}
		\begin{proof}
			We have proved in the lemma \ref{h_1_by_2m_lambda_greater_than_0} that $\frac{1}{2m\lambda}\zeta_G'(\frac{1}{2m \lambda}) < 1$. Now, $\frac{1}{2m\lambda}\zeta_G'(\frac{1}{2m \lambda}) = \zeta_G(\frac{1}{2m \lambda})(\sum_{k = 2}^\infty \frac{\trace(T^k)}{(2m \lambda)^k})$. Note that, $\zeta_G \left(\frac{1}{2m \lambda}\right) > 2$ holds if $\sum_{k = 2}^\infty \frac{\trace(T^k)}{(2m \lambda)^k} < \frac{1}{2}$, which follows from:
			\begin{equation}
			\begin{split}
			\sum_{k = 2}^\infty \frac{\trace(T^k)}{(2m \lambda)^k} & \leq \sum_{k = 2}^\infty \frac{2m \lambda^k}{(2m \lambda)^k} \leq \sum_{k = 2}^\infty \frac{1}{(2m)^{k-1}}\\
			& \leq \frac{1}{2m} + \frac{1}{(2m)^2} + \frac{1}{(2m)^2} + \dots = \frac{1}{2m-1} < \frac{1}{2}.
			\end{split}
			\end{equation}
		\end{proof}
		
		\begin{theorem}\label{zeta_G_x_1_is_equal_to_2} 
			There is a real number $x_1 \in (0, \frac{1}{2m\lambda})$, such that $\zeta_G(x_1) = 2$.  
		\end{theorem}
		\begin{proof}
			From equation (\ref{zeta_with_real_coefficients}) we find $\zeta_G(0) = 1$. Therefore, $\zeta_G(0) - 2 \leq 0$. Also, the lemma \ref{zeta_1_by_2m_lambda_greater_than_2} refers that $\zeta_G(\frac{1}{2m\lambda}) - 2 > 0$. In addition $\zeta_G(x)$ is a monotone function. Therefore, there is a point $x_1 \in (0, \frac{1}{2m\lambda})$ such that $\zeta_G(x_1) = 2$.
		\end{proof}
		
		With the help of the equation (\ref{values_of_denominator}) and the lemma \ref{zeta_G_x_1_is_equal_to_2} we can conclude that 
		\begin{equation}\label{1_by_2m_lambda_less_than_x_0}
		0 < x_1 < \frac{1}{2m \lambda} < x_0 < \frac{1}{\lambda} < 1,
		\end{equation}
		and $\zeta_G(x) > 2$ for all $x > x_1$.
		
		\begin{theorem}\label{zeta_is_admissible}
			The Ihara Zeta function $\zeta_G(x)$ for a given graph $G$ satisfies
			$$R_{\zeta_G} = \frac{1 - a\zeta_G'(a x^\sigma) - a^2 x^\sigma \frac{\sigma}{1 + \sigma}\zeta_G''(ax^\sigma)}{1 - a\zeta_G'(a)} > 0$$ 
			for $x \in [0, 1]$, $a \in (x_1, x_0)$ and $\sigma \in (0, l_\sigma)$ where $x_1$ and $x_0$ are the roots of the equations $\zeta_G(x) = 2$, and $x\zeta'_G(x) = 1$, respectively, in $(0, \frac{1}{\lambda})$. Also, $l_\sigma = \frac{1 - x_1 \zeta'_G(x_1)}{x_1^2 \zeta''_G (x_1)  + x_1 \zeta'_G (x_1) - 1}$.
		\end{theorem}
		
		\begin{proof}
			The equation (\ref{1_by_2m_lambda_less_than_x_0}) suggests that $a \in (x_1, x_0) \subset [0, \frac{1}{\lambda})$. Hence, for any $a \in (x_1, x_0)$ and $x \in [0, 1]$ we have $ax \in [0, x_0)$. The lemma \ref{denominator} along with the equation (\ref{values_of_denominator}) suggests that $1 - a\zeta_G'(a) > 0$ for all $a \in (x_1, x_0)$. 
			
			Note that, for any $\sigma > 0$ and $x \in [0, 1]$ we have $0 \leq x^\sigma \leq 1$. Or $0 \leq a x^\sigma \leq a$ for all $a >0$. As $\zeta_G, \zeta_G'$ and $\zeta_G''$ are monotone increasing functions we have $\zeta'_G(ax^\sigma) \leq \zeta'_G(a)$ and $\zeta''_G(ax^\sigma) \leq \zeta''_G(a)$ for all $\sigma$.
			
			As $1 - a\zeta_G'(a) > 0$ in the range of $a$ we need the range of $\sigma$ such that $1 - a\zeta_G'(a x^\sigma) - a^2 x^\sigma \frac{\sigma}{1 + \sigma}\zeta_G''(ax^\sigma) > 0$, that is $a\zeta_G'(a x^\sigma) + a^2 x^\sigma \frac{\sigma}{1 + \sigma}\zeta_G''(ax^\sigma) \leq 1$. Applying the monotonicity of $\zeta_G'$ and $\zeta_G''$ we find $a\zeta_G'(a x^\sigma) + a^2 x^\sigma \frac{\sigma}{1 + \sigma}\zeta_G''(ax^\sigma) \leq a\zeta_G'(a) + a^2 \frac{\sigma}{1 + \sigma}\zeta_G''(a)$. Now $a\zeta_G'(a) + a^2 \frac{\sigma}{1 + \sigma}\zeta_G''(a) \leq 1$ provides us the limit of $\sigma$. Simplifying we get,
			\begin{equation}
			\begin{split} 
			& \frac{\sigma}{1 + \sigma} < \frac{1 - a\zeta_G'(a)}{a^2 \zeta_G''(a)}\\
			\text{or}~ & \sigma < \frac{1 - a\zeta_G'(a)}{a^2\zeta_G''(a) + a\zeta_G'(a) - 1},
			\end{split} 
			\end{equation}
			for all $a$ such that $x_1 \leq a \leq x_0$. 
			
			Now, lemma \ref{denominator} suggests that $\max_{a \in (x_1, x_0)}[1 - a\zeta_G'(a)] = 1 - x_1\zeta_G'(x_1)$, that is $1 - a\zeta_G'(a) < 1 - x_1\zeta_G'(x_1)$. It indicates
			\begin{equation}
			\begin{split}
			& -(1 - a\zeta_G'(a)) > -(1 - x_1\zeta_G'(x_1))\\
			\text{or}~ & a^2\zeta_G''(a) -(1 - a\zeta_G'(a)) > x_1^2\zeta_G''(x_1) - (1 - x_1\zeta_G'(x_1))\\
			\text{or}~ & a^2\zeta_G''(a) + a\zeta_G'(a) -1 > x_1^2\zeta_G''(x_1) + x_1\zeta_G'(x_1) - 1.\\
			\end{split}
			\end{equation} 
			Combining we get 
			\begin{equation}
			\sigma < \frac{1 - a\zeta_G'(a)}{a^2\zeta_G''(a) + a\zeta_G'(a) - 1} < \frac{1 - x_1 \zeta'_G(x_1)}{x_1^2 \zeta''_G (x_1)  + x_1 \zeta'_G (x_1) - 1} = l_\sigma.
			\end{equation} 
			Therefore, for all $x \in [0,1], a \in (x_1, x_0)$ and $\sigma \in (0, l_\sigma)$ we have $R_{\zeta_G} > 0$. 
		\end{proof}

	\section{Ihara entropy}
		
		Let $R$ be a commutative ring with identity. The ring of formal power series in the variables $x_1, x_2, \dots$ with coefficients in $R$ is denoted by $R\{x_1, x_2, \dots\}$ of. A commutative one-dimensional formal group law over $R$  is a formal power series $\Phi(x, y) \in R\{x, y\}$ of the form $\Phi(x, y) = x + y +$ higher order terms, such that
		\begin{equation}
			\begin{split}
				& \Phi(x, 0) = \Phi(0, x) = x, \\
				& \Phi (\Phi (x, y), z) = \Phi (x, \Phi (y, z)),\\
				& \Phi(x, y) = \Phi(y, x).
			\end{split}
		\end{equation}
		Over a field of characteristic zero, there exists an equivalence of categories between Lie algebras and formal groups. Over the ring $B = \mathbb{Z}[c_1, c_2 , \dots]$ of integral polynomials in infinitely many variables $c_1, c_2, \dots$, we define the formal group logarithm as
		\begin{equation}
			F(s) = \sum_{i = 0}^\infty c_i \frac{s^{i + 1}}{i + 1},
		\end{equation} 
		with $c_0 = 1$. In addition, the formal group exponential is defined as
		\begin{equation}
			G(t) = \sum_{i = 0}^\infty \gamma_i \frac{t^{i + 1}}{i + 1},
		\end{equation}
		such that $F(G(t)) = t$. Recall that, the Lazard formal group law is defined by the formal power series
		\begin{equation}\label{Lazard_formal_group_law} 
			\Phi(s_1, s_2) = G(F(s_1 ) + F(s_2)).
		\end{equation}
		Now, let $\{a_k\}_{k \in \mathbb{N}}$ be a real sequence such that $a_0 \neq 0$ and the power series $\sum_{k = 0}^\infty a_k \frac{t^{k+1}}{k+1}$ converges to a real analytic function $G(t)$, that is $G(t) = \sum_{k = 0}^\infty a_k \frac{t^{k+1}}{k+1}$. Note that, the Hansel's lemma \cite{inverting_formal_power_series} confirms existence of the compositions inverse $F(s)$ such that $F(G(t)) = t$ based on the conditions $a_0 = 0$.		
		
		Let $G$ be a formal power series with its compositional inverse $F$, we define the formal group entropy as 
		\begin{equation}\label{universal_formal_group}
			S([P]) = \sum_{i = 1}^W p_i G\left(\log\left(\frac{1}{p_i}\right)\right),
		\end{equation}
		where $[P] = \{p_i\}_{i = 1}^W$ is a discrete probability distribution.
		
		Recall from equation (\ref{zeta_with_real_coefficients}) that the Ihara zeta function can be expressed as a formal power series. The leading term of $\zeta_G(x)$ is $1$ and $c_1 = 0$. Therefore, the Hensel's lemma suggests that there is no formal power series $F(s)$ such that $F(\zeta_G(t)) = t$. It prevents us to assume $G(t) = \zeta_G(t)$ to define universal formal group entropy (\ref{universal_formal_group}). Hence, for overcoming this obstacle, define a formal power series in terms of Ihara zeta function in a claver fashion.
		
		To introducing it first we transformation $t = \log(\frac{1}{p})$, which refers $p = e^{-t}$. Therefore, the equation (\ref{zeta_with_real_coefficients}) indicates that
		\begin{equation}
			\zeta_G(ae^{-t\sigma}) = 1 + c_2 (ae^{-t\sigma})^2 + c_3 (ae^{-t\sigma})^3 + c_4(ae^{-t\sigma})^4 + \dots.
		\end{equation}
		Also,
		\begin{equation}
			\zeta_G(a) = 1 + c_2 a^2 + c_3 a^3 + c_4 a^4 + \dots.
		\end{equation}
		Hence,
		\begin{equation}
			\zeta_G(ae^{-t\sigma}) - \zeta_G(a) = c_2 a^2(e^{-2t\sigma} - 1) + c_3 a^3(e^{-3t\sigma} - 1) + c_4 a^4(e^{-4t\sigma} - 1) + \dots 
		\end{equation}
		Note that
		\begin{equation}
			e^{-t\sigma} - 1 = -t\sigma + \frac{t \sigma}{2!} - \frac{t \sigma}{3!} + \dots 
		\end{equation}
		Now,
		\begin{equation}
			\begin{split} 
				& \zeta_G(ae^{-t\sigma}) - \zeta_G(a) + 1 - e^{-t\sigma}\\
				= & (1 - e^{-2t\sigma}) + c_2 a^2(e^{-2t\sigma} - 1) + c_3 a^3(e^{-3t\sigma} - 1) + c_4 a^4(e^{-4t\sigma} - 1) + \dots
			\end{split} 
		\end{equation}
		Clearly, $\zeta_G(ae^{-t\sigma}) - \zeta_G(a) + 1 - e^{-2t\sigma}$ has no constant term. The coefficient of $t$ in the power series of $\zeta_G(ae^{-t\sigma}) - \zeta_G(a) + 1 - e^{-2t\sigma}$ is 
		\begin{equation}
			\begin{split}
				& \frac{d}{dt}\left[\zeta_G(ae^{-t\sigma}) - \zeta_G(a) + 1 - e^{-t\sigma}\right]|_{t = 0}\\
				= & \left[-a \sigma \zeta_G'(ae^{-t\sigma}) e^{-t\sigma} + \sigma e^{-t\sigma} \right] |_{t = 0} = \sigma \left[1 - a \zeta'_G(a)\right].
			\end{split}
		\end{equation}
		Hence coefficient of $t$ in the formal power series of the following expression is unity:
		\begin{equation}
			\frac{\zeta_G(ae^{-t\sigma}) - \zeta_G(a) + 1 - e^{-t\sigma}}{\sigma \left(1 - a \zeta'_G(a)\right)} = \frac{1}{(1 - a \zeta'_G(a))}\left[\frac{\zeta_G(ae^{-t\sigma}) - \zeta_G(a) + 1 - e^{-t\sigma}}{\sigma} \right].
		\end{equation} 
		Therefore, we have a formal power series associated to the Ihara zeta function satisfying the requirements of Hensel's lemma, which is given by
		\begin{equation}
			G(t) = \frac{1}{(1 - a \zeta'_G(a))}\left[\frac{\zeta_G(ae^{-t\sigma}) - e^{-t\sigma}}{\sigma} + \frac{1 - \zeta_G(a)}{\sigma} \right]. 
		\end{equation}
		This particular choice of $G$ will be clarified to the reader when we discuss information theoretic aspects of the Ihara entropy in the next section. Now replacing $t = \log(\frac{1}{p})$ we find 
		\begin{equation}
			G\left(\log\left(\frac{1}{p}\right)\right) = \frac{1}{(1 - a \zeta'_G(a))}\left[\frac{\zeta_G(ap^\sigma) - p^\sigma}{\sigma} + \frac{1 - \zeta_G(a)}{\sigma} \right].
		\end{equation}
		It leads us to construct the formal group theoretic entropy associated to the Ihara zeta function, which is defined below. 
		\begin{deffinition}\label{trace_from_entropy}
			Given a graph $G$ the Ihara entropy of a discrete probability distribution $[\mathcal{P}] = \{p_1, p_2, \dots p_W\}$ is defined by
			$$S_G([\mathcal{P}]) = \frac{1}{1 - a\zeta_G'(a)}\sum_{i = 1}^W p_i\left(\frac{\zeta_G(a p_i^\sigma) - p_i^\sigma}{\sigma} + \frac{1 - \zeta_G(a)}{\sigma}\right),$$
			where $a$ and $\sigma$ are two parameters with $a \in (x_1, x_0)$ and $\sigma \in (0, l_\sigma)$. Here $x_1$ and $x_0$ are roots of the equations $\zeta(x) = 2$, and $x\zeta'_G(x) = 1$, respectively. In addition, $l_\sigma = \min\left\{1, \frac{1 - a \zeta'_G(a)}{a^2 \zeta''_G \left( \frac{1}{\lambda}\right)  + a \zeta'_G \left( \frac{1}{\lambda}\right) - 1}\right\}$.
		\end{deffinition}

	\section{Information theoretic properties}
	
		The Ihara entropy is mentioned in the definition \ref{trace_from_entropy} fulfills the Shannon-Khinchin axioms \cite{shannon1963mathematical, shannon1948mathematical, khinchin2013mathematical}, relevant in the information theoretic context. The Shannon-Khinchin axioms for an entropy $S([\mathcal{P}])$ are mentioned below:
		\begin{enumerate}
			\item \label{axiom_1} 
				The function $S([\mathcal{P}])$ is continuous with respect to all its arguments $p_i$, where $[\mathcal{P}] = \{p_i\}_{i = 1}^W$ is a discrete probability distribution.
			\item \label{axiom_2}
				The function $S([\mathcal{P}])$ is maximum for the uniform distribution $[\mathcal{P}] = \{\frac{1}{W}\}_{i = 1}^W$.
			\item \label{axiom_3}
				Adding a zero probability event to a probability distribution does not alter its entropy, that is $S([\mathcal{P}_1]) = S([\mathcal{P}])$ where $[\mathcal{P}_1] = \{p_i\}_{i = 1}^W \cup \{0\}$.
			\item \label{axiom_4}
				Given two subsystems $A, B$ of a statistical system, $S(A + B) = S(A) + S(B|A)$.
		\end{enumerate}
		
		Given a probability value $p$ define a function $s:[0, 1] \rightarrow \mathbb{R}$ such that 
		\begin{equation}
			s(p) = \frac{1}{\sigma(1 - a\zeta_G'(a))} \left[p \left\{ \zeta_G(a p^\sigma) - p^\sigma + 1 - \zeta_G(a) \right\}\right]. 
		\end{equation}
		Therefore, the Ihara entropy $S_G([\mathcal{P}]) = \sum_{i = 1}^W s(p_i)$. 
		
		Clearly, $s(p)$ is a continuous function of $p$. Therefore, the Ihara entropy $S_G([\mathcal{P}])$ is continuous with respect to all its arguments $p_i$ for $i = 1, 2, \dots W$. Therefore Ihara entropy satisfies the axiom \ref{axiom_1}.
			
		The axiom \ref{axiom_3} also trivially satisfied as $s(0) = 0$, that is $0$ probability alters nothing in $S([\mathcal{P}])$.

		The axiom \ref{axiom_2} is concerned to the maximization of the $s(p)$ in $(0,1)$, which is ensured by the following theorem.
		
		\begin{theorem}
			There exists a maxima of $s(p)$ in $(0, 1)$ that is there is a point $c \in (0, 1)$, such that, $s'(c) = 0$ and $s''(c) < 0$.
		\end{theorem}
		\begin{proof}  
			We have
			\begin{equation}
				\begin{split}
					s'(p) & = \frac{1}{\sigma(1 - a\zeta_G'(a))} \left[\zeta_G(a p^\sigma) - p^\sigma + 1 - \zeta_G(a) + p\{\zeta'_G(a p^\sigma)a\sigma p^{\sigma - 1} - \sigma p^{\sigma-1} \}\right] \\
					& = \frac{1}{\sigma(1 - a\zeta_G'(a))} \left[ a \sigma p^\sigma \zeta_G'(a p^\sigma) + \zeta_G(a p^\sigma) - (1 + \sigma)p^\sigma - \zeta_G(a) + 1\right].
				\end{split}
			\end{equation}
			Now $s'(p) = 0$ indicates that the function 
			\begin{equation}
				h(p) = a \sigma p^\sigma \zeta_G'(a p^\sigma) + \zeta_G(a p^\sigma) - (1 + \sigma)p^\sigma - \zeta_G(a) + 1.
			\end{equation}
			has a root in $(0,1)$. Now $h(0) = 2 - \zeta_G(a)$ and $h(1) = \sigma + a \sigma \zeta_G'(a)$. In definition \ref{trace_from_entropy} we consider $a > x_1$. Recall that the lemma \ref{zeta_G_x_1_is_equal_to_2} indicates $\zeta_G(x) > 2$ for all $a > x_1$. That is $h(0) < 0$. Also $h(1) > 0$. As $h$ is a continuous function of $p$ there is at least one point $p = c$ in $(0, 1)$ such that $h(c) = 0$ that is $s'(c) = 0$. Now we prove $s''(p) < 0$ for all $p \in (0, 1)$ that is $s''(c) < 0$. Note that, 
			\begin{equation}
				\begin{split}
					s''(p) = \frac{1}{\sigma(1 - a\zeta_G'(a))} [ a \sigma^2 p^{\sigma - 1} \zeta_G'(a p^\sigma)  + a \sigma p^\sigma \zeta_G''(a p^\sigma) a \sigma p^{\sigma -1} & \\
					+ \zeta_G'(a p^\sigma) a \sigma p^{\sigma - 1} - (1 + \sigma) p^{\sigma - 1} ] &\\
					\text{or}~ s''(p) = \frac{1}{\sigma(1 - a\zeta_G'(a))} [a^2 \sigma^2 p^{2\sigma - 1} \zeta_G''(a p^\sigma) + a \sigma (\sigma + 1) p^{\sigma - 1} \zeta_G'(a p^\sigma)&\\
					- \sigma(1 + \sigma) p^{\sigma - 1}]&\\
					\text{or}~ s''(p) = \frac{(\sigma + 1) p^{\sigma-1}}{1 - a\zeta_G'(a)}[a^2\frac{\sigma}{\sigma + 1} p^\sigma \zeta_G''(a p^\sigma) + a \zeta_G'(a p^\sigma) - 1] &
				\end{split}
			\end{equation}
			The theorem \ref{zeta_is_admissible} indicates that in the specific range of $a$ and $\sigma$ we have 
			\begin{equation}
				\frac{1 - a\zeta_G'(a x^\sigma) - a^2 x^\sigma \frac{\sigma}{1 + \sigma}\zeta_G''(ax^\sigma)}{1 - a\zeta_G'(a)} > 0.
			\end{equation}
			Clearly $s''(p) < 0$ for all $p \in (0, 1)$. Hence $s''(c) < 0$.
		\end{proof} 
		
		Recall the definition \ref{trace_from_entropy}. The entropy $S_G([\mathcal{P}])$ consider the maximum value if $s(p_i)$ is maximum for all $p_i \in [\mathcal{P}]$. Thus, to maximize $S_G([\mathcal{P}])$ we need $p_i = c$ for all $i$, which is the uniform distribution after a normalization. Therefore, the Ihara entropy mentioned in definition \ref{trace_from_entropy} fulfills the axiom \ref{axiom_2} of the Shannon-Khinchin conditions.
	
		The axiom \ref{axiom_4} is generalised by the composition law of the Lazard formal group law mentioned in equation  (\ref{Lazard_formal_group_law}) \cite[theorem 1]{tempesta2016beyond}. 
	
		\begin{corollary}
			The Ihara entropy $S_G([P])$ becomes the Shannon entropy when $\sigma \rightarrow 0$.
		\end{corollary}
	
		\begin{proof}
			Note that,
			\begin{equation}
				\frac{d}{d\sigma} p \left\{ \zeta_G(a p^\sigma) - p^\sigma + 1 - \zeta_G(a) \right\} = (\zeta_G'(ap^\sigma) - 1)p\log(p).
			\end{equation}
			Also, $\lim_{\sigma \rightarrow 0}(\zeta_G'(ap^\sigma) - 1) = (\zeta_G'(a) - 1)$. Now applying the LH{\^o}pital's rule we reach to the conclusion.
		\end{proof}

	\section{Ihara entropy in billiard dynamics}
	
		The interface of zeta functions and the dynamical system is a well-studied topic in mathematics \cite{pollicott2001dynamical}. The Ihara zeta function of a finite graph can be considered as a Ruelle type zeta function associated to a symbolic dynamics defined on the graph \cite{nikitin2001ihara}.
		
		Recall the set of all oriented edges $\mathcal{E}$ mentioned in equation (\ref{alphabet}). To define a symbolic dynamical system, consider the set $\mathcal{E}$ as a collections of symbols, that is an alphabet \cite{lind1995introduction}. Two edges $e^{(i)}$ and $e^{(j)} \in \mathcal{E}$ are composeble if the terminal vertex of $e^{(i)}$ is the initial vertex of $e^{(j)}$ and the composition is written as $e^{(i)}e^{(j)}$. A bi-infinite walk $w$ is a bi-infinite sequence of edges from $\mathcal{E}$ given by
		\begin{equation}
		w = \prod_{i \in \mathbb{Z}} e_i = \dots e_{-2} e_{-1} e_0 e_1 e_2 \dots,
		\end{equation}
		such that any two constitutive edges $e_i$ and $e_{i + 1}$ are composeble for all $i \in \mathbb{Z}$. In symbolic dynamics we are interested in studying the bi-infinite sequences of symbols that is the bi-infinite walks over $\mathcal{E}$. The set of all such walks is called full $\mathcal{E}$-shift, which is denoted and precisely defined by 
		\begin{equation}
		\mathcal{E}^{\mathbb{Z}} = \{w: w = \prod_{i \in \mathbb{Z}} e_i, ~e_i \in \mathcal{E}, ~\text{for all}~ i \in \mathbb{Z}\}.
		\end{equation}
		
		A block over $\mathcal{E}$ is a finite walk $Q = e_1e_2\dots e_k$ over the edges in $\mathcal{E}$. The length of the block $Q$ is denoted by $\gamma(Q)$ which is the number of edges in $Q$. In graph theoretical nomenclature, $\gamma(Q)$ is length of the walk $Q$. If $w$ is an element in $\mathcal{E}^{\mathbb{Z}}$ then the block of coordinates in $w$ from position $i$ to $j$ is represented by
		\begin{equation}
		w_{[i,j]} = \prod_{k = i}^j e_k = e_i e_{i + 1} \dots e_j. 
		\end{equation} 
		If $w \in \mathcal{E}^{\mathbb{Z}}$ and $Q$ is a block over $\mathcal{E}$, we say that $Q$ occurs in $w$ if there are indices $i$ and $j$ so that $W = w_{[i,j]}$. 
		
		Now we define a set of forbidden blocks $\mathcal{F}$ over $\mathcal{A}$.
		\begin{equation}
		\mathcal{F} = \{e^{(i)}(e^{(i)})^{-1}: e^{(i)} \in \mathcal{E}\}.
		\end{equation} 
		For any such $\mathcal{F}$, define $X_{\mathcal{F}}$ to be the subset of sequences in $\mathcal{E}^\mathbb{Z}$ which does not contain the blocks in $\mathcal{F}$. A shift space is a subset $X$ of a full shift $\mathcal{E}^\mathbb{Z}$, such that $X = X_\mathcal{F}$ for some collection $\mathcal{F}$ of forbidden blocks over $\mathcal{A}$. Note that, $\mathcal{F}$ is finite and hence $X_\mathcal{F}$ is a shift space of finite type and of one-step.
		
		The edges of $\overline{G}$ represents the composition of two elements in $\mathcal{E}$. Therefore, a fine length walk in $\overline{G}$ represents a block over $\mathcal{E}$ and a prime is a closed cycle in $\overline{G}$. Also, there is no cycle of length $2$ in the graph $\overline{G}$.
		
		Two primes $P^{(1)} = e^{(1)}_1 e^{(1)}_2 \dots e^{(1)}_r$ and $P^{(2)} = e^{(2)}_1 e^{(2)}_2 \dots e^{(2)}_s$ are composable if the terminal vertex of $e^{(1)}_r$ is the initial vertex of $e^{(2)}_1$, and the composition is denoted by $P^{(1)}P^{(2)} = e^{(1)}_1 e^{(1)}_2 \dots e^{(1)}_re^{(2)}_1 e^{(2)}_2 \dots e^{(2)}_s$. In a similar fashion, we can define product of a prime $P$ and a finite walk $Q$, or product of two finite walks $Q_1$ and $Q_2$.
		
		\begin{lemma}
			Any infinite walk $w \in \mathcal{E}^{\mathbb{Z}}$ can be expressed as a product of primes and walks of finite length.
		\end{lemma}
		
		\begin{proof}
			Note that, an infinite product of primes and finite walks is an infinite walk. The set of all symbols $\mathcal{E}$ contains $2m$ symbols where $m$ is finite. Therefore, at least one edge must be repeated in $w$. Every edge has an initial and a terminal vertex and the number of the vertex is also finite. Hence, the vertices must be repeated in the infinite walk $w$. For simplicity let the initial vertex $v_0 = i(e_0)$ is repeated after $k$ edges, that is $v_0 = i(e_0) = t(e_k)$. Now consider the block $e_0e_1\dots e_k$ which is a cycle and must be contained in a prime. Let there is a vertex $v$ which appears exactly once in the walk $w$. Therefore, there are two edges $e_{k}$ and $e_{k + 1}$ in $w$ such that $t(e_k) = v = i(e_{k + 1})$, which appear in walk $w$ only once. Therefore there is a finite walk $\dots e_k e_{k + 1} \dots$ which appears in $w$ only once. It constructs a finite walk. Hence the proof.  
		\end{proof}
		
		The above lemma indicates that any bi-infinite walk $w \in \mathcal{E}^{\mathbb{Z}}$ can be expressed as 
		\begin{equation}
		w = \prod_{i \in \mathbb{Z}} P_{i}^{p_i} Q_i,
		\end{equation}
		where $p_i$ is the number of consecutive repetitions of the prime $P_i$ and $Q_i$ is the finite walk after the repetition of prime $P_i$. In general, length of $Q_i$ may be $0$. Note that, a walk $w$ may have multiple representations.
		
		Note that, the oriented line graph $\overline{G}$ preserves the sequence of walks on the graph $G$. As an adjacency matrix, the $(i,j)$-th element of $T^k$ denotes the total number of walks of length $k$ between vertices $i$ and $j$. In addition, the diagonal elements of $T^k$ denoted the length of cycles of length $k$. When $k \rightarrow \infty$ we get biinfinite walks in the graph. As $\overline{G}$ is a simple connected graph the matrix $T$ is an irreducible matrix which acts as the Perron-Frobenius operator associated with the dynamical system. In this context, the Ihara entropy can be considered as the entropy of the symbolic dynamical system associated with the graph.
		
		Now we consider a particular symbolic dynamical system: the dynamics of billiards \cite{dahlqvist1995approximate, morita1991symbolic}. In this article, we introduced a graph-theoretic billiard system. Consider a set of reflectors as a set of vertices in a graph. The shape and size of the reflectors may be arbitrarily chosen. The reflectors are arranged in a plane. Between two vertices there is an edge if a billiard can be reflected between two corresponding reflectors. Given any graph, we can design such a system of reflectors by placing the vertices on a plane. Edges will be drawn with straight lines such that no two edges coincide. 
		
		\begin{figure}
			\centering 
			\begin{subfigure}{0.4\textwidth}
				\centering 
				\begin{tikzpicture}
				\draw (0, 0) circle [radius= .5cm];
				\node at (0, 0) {$1$};
				\draw (2, 0) circle [radius= .5cm];
				\node at (2, 0) {$2$};
				\draw (0, 2) circle [radius= .5cm];
				\node at (0, 2) {$3$};
				\draw (2, 2) circle [radius= .5cm];
				\node at (2, 2) {$4$};
				\draw (1, 1) circle [radius= .5cm];
				\node at (1, 1) {$5$};
				\end{tikzpicture}
				\caption{System of reflectors} 
				\label{system_of_reflectors} 
			\end{subfigure}
			\begin{subfigure}{.4\textwidth}
				\centering 
				\begin{tikzpicture}
					\draw [fill] (0, 0) circle [radius= .05];
					\node [below] at (0, 0) {$1$};
					\draw [fill] (2, 0) circle [radius= .05];
					\node [below] at (2, 0) {$2$};
					\draw [fill] (0, 2) circle [radius= .05];
					\node [above] at (0, 2) {$3$};
					\draw [fill] (2, 2) circle [radius= .05];
					\node [above] at (2, 2) {$4$};
					\draw [fill] (1, 1) circle [radius= .05];
					\node [right] at (1, 1) {$5$};
					\draw (0, 0) -- (0, 2) -- (2, 2) -- (2, 0) -- (0, 0) -- (1, 1) -- (2, 2);
					\draw (2, 0) -- (1, 1) -- (0, 2);
				\end{tikzpicture}
				\caption{Representation with a graph} 
				\label{graph_representation}
			\end{subfigure} 
		
			\begin{subfigure}{1\textwidth}
				\centering 
				\includegraphics[scale = .5]{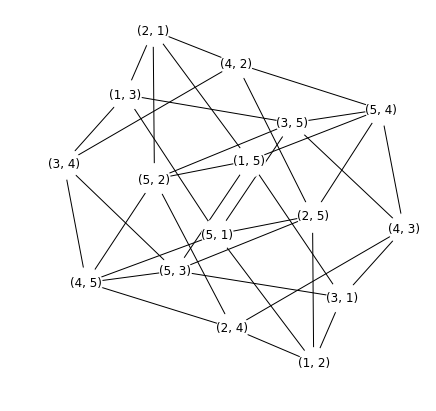}
				\caption{Networkx \cite{hagberg2008exploring} generated oriented line graph}
				\label{oriented_graph}
			\end{subfigure}
		\end{figure}
		
		For instance, consider figure \ref{system_of_reflectors} which is an arrangement of perfectly reflecting circles in the plane which is often studied in Sinai billiard dynamics \cite{Weisstein_Sinai}.This system can be represented as the combinatorial graph depicted in figure \ref{graph_representation}. A billiard can be reflected at any vertex and move towards a neighboring vertex. The movement of billiard over the edge $(u,v)$ may be from $u$ to $v$ or from $v$ to $u$. Therefore we can define two orientations on every edge. Hence, we find a total $16$ symbols over edges of the above graph mentioned below:
		\begin{equation}
		\begin{matrix}
		e^{(1)} = (1, 2), & e^{(2)} = (2, 4), & e^{(3)} = (4, 3), & e^{(4)} = (3, 1),\\
		e^{(5)} = (1, 5), & e^{(6)} = (2, 5), & e^{(7)} = (3, 5), & e^{(8)} = (4, 5),
		\end{matrix}
		\end{equation}
		and $e^{(8 + k)} = (e^{(k)})^{-1}$ for $k = 1, 2,\dots 8$. The set of all symbols is $\mathcal{E} = \{e^{(k)}: k = 1, 2, \dots 16\}$. An example of biinfinite sequence can be given by
		\begin{equation}
		\begin{split} 
		W = & \dots e^{(1)} e^{(6)}(e^{(5)})^{-1}e^{(1)} e^{(6)}(e^{(5)})^{-1}e^{(1)} e^{(2)}e^{(3)}e^{(7)} (e^{(8)})^{-1}e^{(3)}e^{(7)} (e^{(8)})^{-1} \dots \\
		= & (e^{(1)} e^{(6)}(e^{(5)})^{-1})^{p_1}e^{(1)} e^{(2)}(e^{(3)}e^{(7)} (e^{(8)})^{-1})^{p_2} ~\text{where}~ p_1, p_2 \in \mathbb{N}. 
		\end{split} 
		\end{equation}
		Here $W = (P_1)^{p_1} Q (P_2)^{p_2}$ where $P_1 = e^{(1)} e^{(6)}(e^{(5)})^{-1}$, $P_2 = e^{(3)}e^{(7)} (e^{(8)})^{-1}$ are two prime cycles and $Q = e^{(1)} e^{(2)}$ is a finite walk. The oriented line graph is depicted in the figure \ref{oriented_graph} which has $16$ nodes.
		
		As the equation (\ref{Ihara_zeta_1}) and (\ref{Ihara_zeta_2}) are not calculation friendly, we derive the Ihara zeta function using the Ihara determinant formula, which is
		\begin{equation}
		\frac{1}{\zeta_G(x)} = (1 - x^2)^{m - n}\det(I - Ax + (D - I)x^2),
		\end{equation}
		where $0 \leq x < \frac{1}{\lambda}$, $A$, and $D$ are the adjacency matrix, and degree matrix of graph $G$, as well as $I$ is the identity matrix of order $n$. The Ihara zeta function for the graph depicted in figure \ref{graph_representation} is 
		\begin{equation}
		\zeta_G(x) = \frac{1}{\left(1-x^2\right)^3 \left(48 x^{10}+32 x^8-32 x^7-8 x^6-32 x^5-4 x^4-8 x^3+3 x^2+1\right)}
		\end{equation}
		Hence for any probability value $p$, we have $\zeta_G(ap^\sigma) = \frac{1}{A}$ where
		\begin{equation}
		\begin{split}
		A = \left(1-a^2 p^{2 \sigma }\right)^3 [48 a^{10} p^{10 \sigma }+32 a^8 p^{8 \sigma }-32 a^7 p^{7 \sigma }-8 a^6 p^{6 \sigma }-32 a^5 p^{5 \sigma } & \\
		-4 a^4 p^{4 \sigma }-8 a^3 p^{3 \sigma }+3 a^2 p^{2 \sigma }+1]&.
		\end{split} 
		\end{equation}
		Also, for any $a$ we have
		\begin{equation}
		\zeta_G(a) = \frac{1}{\left(1-a^2\right)^3 \left(48 a^{10}+32 a^8-32 a^7-8 a^6-32 a^5-4 a^4-8 a^3+3 a^2+1\right)}.
		\end{equation}
		The first derivative
		\begin{equation}
		\begin{split} 
		\zeta_G'(x) = & \frac{6 x}{\left(1-x^2\right)^4 \left(48 x^{10}+32 x^8-32 x^7-8 x^6-32 x^5-4 x^4-8 x^3+3 x^2+1\right)}\\
		& -\frac{480 x^9+256 x^7-224 x^6-48 x^5-160 x^4-16 x^3-24 x^2+6 x}{\left(1-x^2\right)^3 \left(48 x^{10}+32 x^8-32 x^7-8 x^6-32 x^5-4 x^4-8 x^3+3 x^2+1\right)^2}.
		\end{split} 
		\end{equation}
		It gives us the expression of $\frac{1}{\sigma(1 - a\zeta_G'(a))}$, where $\sigma(1 - a\zeta_G'(a)) = $
		\begin{equation}
		\begin{split}
		\sigma -\frac{6 a^2 \sigma }{(a-1)^5 (a+1)^4 \left(2 a^2+1\right)^2 \left(12 a^5+12 a^4+8 a^3-a-1\right)}& \\
		-\frac{2 \left(120 a^6+4 a^4-56 a^3-14 a^2-12 a+3\right) a^2 \sigma }{(a-1)^5 (a+1)^3 \left(2 a^2+1\right)^3 \left(-12 a^5-12 a^4-8 a^3+a+1\right)^2}&
		\end{split}
		\end{equation}
		Now for any probability $P$ we have $s(p) = p\frac{\zeta_G(ap^\sigma)-\zeta_G(a)+1-p^\sigma}{\sigma(1 - a\zeta_G'(a))} = \frac{p\times A}{B}$ where
		\begin{equation}
		\begin{split}
		A = 1 - p^\sigma + \frac{1}{(a-1)^4 (a+1)^3 \left(2 a^2+1\right)^2 \left(6 a^3+a-1\right) (2 a (a+1)+1)} &\\
		-\frac{1}{\left(a p^{\sigma }-1\right)^4 \left(a p^{\sigma }+1\right)^3 \left(2 a^2 p^{2 \sigma }+1\right)^2 \left(6 a^3 p^{3 \sigma }+a p^{\sigma }-1\right) \left(2 a p^{\sigma } \left(a p^{\sigma }+1\right)+1\right)} & \\
		\text{and}~ B = \sigma -\frac{8 a^3 (a (a (2 a (a (24 a (a+1)+11)-2)-13)-8)-3) \sigma }{(a-1)^5 (a+1)^4 \left(2 a^2+1\right)^3 \left(6 a^3+a-1\right)^2 (2 a (a+1)+1)^2}.& 
		\end{split}
		\end{equation}
		The Ihara entropy of the billiard system is $S_G([\mathcal{P}]) = \sum_i s(p_i)$, where $\mathcal{P} = \{p_i\}_{i = 1}^W$ is any given discrete probability distribution.

	\section*{Acknowledgement}
		
		SD is thankful to his Ph.D. supervisor Dr. Subhashish Banerjee for introducing him with the Ihara Zeta function and its applications in quantum information theory.
		

\end{document}